\documentclass{amsart}

\usepackage{amsmath}
\usepackage{amssymb}
\usepackage{amsfonts}
\usepackage{booktabs} 

\usepackage{alltt}
\usepackage{bm}
\usepackage{mathrsfs}
\usepackage{graphicx}
\usepackage{color}
\usepackage{mathtools}
\usepackage[T1]{fontenc}
\usepackage[utf8]{inputenc}
\usepackage[export]{adjustbox}
\usepackage[colorlinks]{hyperref}

\usepackage{graphicx} 
	\usepackage{pgf,tikz,pgfplots}
\usetikzlibrary{arrows}
\newcommand{\fsize}{.7}

\usepackage{caption} 

\usepackage[left=1.72in,right=1.72in,top=1.3in,bottom=1.3in]{geometry}

\newcommand{\R}{\mathbb R}
\newcommand{\Z}{\mathbb Z}

\newcommand{\Exp}[1]{\exp\left( #1\right)}

\newcommand{\supp}{\operatorname{supp}}

\newcommand{\spann}{\operatorname{span}}

\newcommand{\sou}{M}
\newcommand{\Jg}{J_\mathrm{good}}
\newcommand{\Jr}{J_\mathrm{bad}}
\newcommand{\Ng}{N_\mathrm{good}}
\newcommand{\Nr}{N_\mathrm{bad}}

\numberwithin{equation}{section}

\numberwithin{equation}{section}

\newtheorem{theorem}{Theorem} 

\newtheorem{lemma}[theorem]{Lemma}

\newtheorem*{thm*}{Theorem}
\newtheorem*{lemma*}{Lemma}
\newtheorem*{prop*}{Proposition}
\newtheorem*{assumption}{Assumption A}

\theoremstyle{definition}

\newtheorem*{dfn*}{Definition}

\theoremstyle{remark}

\makeatletter
\def\@setauthors{%
  \begingroup
  \def\thanks{\protect\thanks@warning}%
  \trivlist
  \centering\large \@topsep30\p@\relax
  \advance\@topsep by -\baselineskip
  \item\relax
  \author@andify\authors
  \def\\{\protect\linebreak}%
  \authors%
  \ifx\@empty\contribs
  \else
    ,\penalty-3 \space \@setcontribs
    \@closetoccontribs
  \fi
  \endtrivlist
  \endgroup
}
\def\@settitle{\begin{center}%
  \baselineskip14\p@\relax
    \normalfont\LARGE
  \@title
  \end{center}%
}
\makeatother

\pgfplotsset{compat=1.13}

\begin{document}

\title[Perfect reconstructions for multiple sources]{Perfect partial reconstructions for multiple simultaneous sources}

\author[Wittsten]{Jens Wittsten}
\address[Jens Wittsten]{Center for Mathematical Sciences, Lund University, Sweden}
\email{jensw@maths.lth.se}

\author[Andersson]{Fredrik Andersson}
\address[Fredrik Andersson]{Seismic Apparition GmbH, Zurich, Switzerland \&
Institute of Geophysics, ETH-Zurich, Zurich, Switzerland}
\email{fandersson@seismicapparition.com}

\author[Robertsson]{Johan Robertsson}
\address[Johan Robertsson]{Seismic Apparition GmbH, Zurich, Switzerland \&
Institute of Geophysics, ETH-Zurich, Zurich, Switzerland}
\email{jrobertsson@seismicapparition.com}

\author[Amundsen]{Lasse Amundsen}
\address[Lasse Amundsen]{Statoil Research Centre, Trondheim, Norway \&
Department of Geoscience and Petroleum, NTNU, Trondheim, Norway}
\email{lam@statoil.com}

\begin{abstract}
A major focus of research in the seismic industry of the past two decades has been the acquisition and subsequent separation of seismic data using multiple sources fired simultaneously. The recently introduced method of {\it signal apparition} provides a new take on the problem by replacing the random time-shifts usually employed to encode the different sources by fully deterministic periodic time-shifts.
In this paper we give a mathematical proof showing that the signal apparition method results in optimally large regions in the frequency-wavenumber space where exact separation of sources is achieved. These regions are diamond-shaped and we prove that using any other method of source encoding results in strictly smaller regions of exact separation. The results are valid for arbitrary number of sources. Numerical examples for different number of sources (three resp.~four sources) demonstrate the exact recovery of these diamond-shaped regions. The theoretical proofs'  implementation in the field is illustrated by the results of a conducted field test. 
\end{abstract}


\maketitle

\section{Introduction}

Methods for simultaneous source separation have been a major focus in the seismic industry over the last two decades \cite{beasley1998new}.  Acquiring seismic data without having to wait for the response of one source to be recorded before exciting one or more sources at other shot points promises to radically increase productivity.  This can be essential for instance to make complex wide azimuth seismic surveys cost-effective.  Other constraints such as completing a survey within time-share agreements or in between fish spawning seasons can also greatly benefit from a significant increase in productivity.

The simultaneous source problem is fundamentally an ill-posed problem above a certain frequency and to solve the problem it is necessary to introduce additional constraints \cite{andersson2017flawless}.  A popular method used in industry is based on the science of compressive sensing.  By using random time dithers when exciting sources relative to other sources being excited, it is possible to invert the recorded seismic data for individual source responses under assumptions such as coherency and sparseness of seismic data, see 
\cite{akerberg2008simultaneous},
\cite{berkhout2008changing},
\cite{hampson2008acquisition},
\cite{ikelle2010coding},
\cite{moore2010simultaneous},
\cite{wapenaar2012deblending},
\cite{chen2014iterative},
\cite{langhammer2015triple},
\cite{mueller2015benefit},
\cite{andersson2016deblending},
and references therein.  
The recently introduced concept of signal apparition offers a fundamentally different approach to solve the source separation problem \cite{robertsson2016signal}.  Instead of using random dithers, deterministic periodic variations of excitation times of a source relative to other sources result in the mapping of data into multiple signal cones away from the usual signal cone centered at wavenumber zero and bounded by the propagation velocity of the recording medium ($1500\, \mathrm{m}/\mathrm{s}$ in the case of marine seismic data). Andersson et al.~\cite{andersson2017flawless} showed that for two sources simultaneously acquiring data along two lines over a general 3D heterogeneous sub-surface, the apparition-style acquisition strategy results in a region in the frequency-wavenumber space where the separation of sources is exact which is twice as large as what is possible to achieve using random dithers for instance. Andersson et al.~\cite{andersson2017flawless} referred to these regions as “flawless diamonds” and demonstrated that exact separation of sources in these regions is unique to signal apparition.
Outside the flawless diamonds, the simultaneous source problem is solved using additional constraints to tackle the otherwise ill-posed problem \cite{andersson2016seismic}, \cite{andersson2017analytic}.

In this paper we generalize the findings of Andersson et al.~\cite{andersson2017flawless} to that of $\sou$ simultaneous sources \cite{andersson2017multisource}, \cite{amundsen2017multisource}.  We show that the method of signal apparition results in optimally large regions in the frequency-wavenumber space for exact separation of sources.  Moreover, we show that all other methods for source encoding results in regions of exact separation which are smaller than that obtained by encoding the sources using signal apparition.  The main part of the paper contains the proofs of two theorems that demonstrate that signal apparition is unique and optimal in the sense of exactly separating the response of $\sou$ sources.  The theorems are supported by lemmas included in the appendix.  Following the theory section we present a numerical example, as well as the results of a field test which describes the implementation of the theory in a practical environment.

\section{Theory}\label{section:theory}
In this section we prove the above mentioned theoretic implications of using signal apparition to encode sources during simultaneous source acquisition, namely, that signal apparition optimizes the area of exact separation of the responses of simultaneous sources (Theorem \ref{maintheorem2}), and it is the only method to do so (Theorem \ref{maintheorem}). 
We consider a marine environment in which $\sou $ simultaneous sources are acquired along a single line over a 
complex 3D sub-surface. We remark that this type of source modulation along a line
by no means limits the application to 2D data, and it is straightforward to generalize the results presented in this paper to source modulation in a plane. Note also that, in the case of 3D towed streamer surveys, the anticipated shot timing and position errors that can occur due, e.g., to feathering, will have a negligible impact on the results, and we refer the readers to Wittsten et al.~\cite{wittsten2018perturbations} for an in depth discussion on the influence of perturbations on the methodology.

We first make an observation. Consider an experiment where a source is fired on equidistant shot positions, spaced a distance $\triangle_x$ apart along a line (of infinite length), and recorded on a stationary receiver.  This spatial sampling frequency corresponds to a Nyquist wavenumber of $k_S=1/(2\triangle_x)$, and the slowest possible apparent velocity in the recordings is identical to the propagation velocity of the recording medium. This results in a maximum frequency (Nyquist frequency), depending on the spatial sampling interval $\triangle_x$, below which all energy is unaliased. Moreover, after a temporal and spatial Fourier transform ($\omega  {k_x}$), all signal energy is confined to a ``signal cone'' bounded by the sound speed of the recording medium. This also means that large parts of the $\omega  {k_x}$-spectrum inside the Nyquist frequency and wavenumber are zero.

\definecolor{ffqqqq}{rgb}{.6,0.,0.}
\definecolor{qqffqq}{rgb}{0.,.6,0.}
\definecolor{qqqqff}{rgb}{0.,0.,.6}
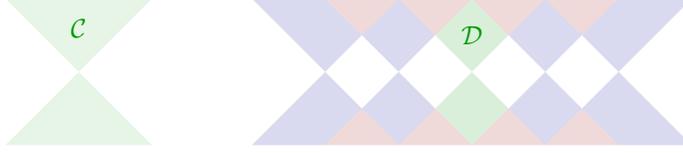
\begin{figure}
	\centering
	\begin{tikzpicture}[line cap=round,line join=round,x=.195cm,y=.195cm]
	\clip(-10.6,-5.3) rectangle (10.6,5.3);
	\fill[color=qqffqq,fill=qqffqq,fill opacity=0.1] (0.,0.) -- (-5.,5.) -- (5.,5.) -- cycle;
	\fill[color=qqffqq,fill=qqffqq,fill opacity=0.1] (0.,0.) -- (5.,-5.) -- (-5.,-5.) -- cycle;
	\draw (0,3) node[qqffqq]{$\mathcal{C}$};
	\end{tikzpicture}
	\begin{tikzpicture}[line cap=round,line join=round,x=.195cm,y=.195cm]
	\clip(-15.6,-5.3) rectangle (15.6,5.3);
	\fill[color=qqffqq,fill=qqffqq,fill opacity=0.15] (0.,0.) -- (-2.5,2.5) -- (0.,5) -- (2.5,2.5) -- cycle;
	\fill[color=qqffqq,fill=qqffqq,fill opacity=0.15] (0.,0.) -- ( 2.5,-2.5) -- (0.,-5) -- (-2.5,-2.5) -- cycle;
	\fill[color=qqqqff,fill=qqqqff,fill opacity=0.15] (5.,0.) -- (2.5,2.5) -- (5.,5) -- (7.5,2.5) -- cycle;
	\fill[color=qqqqff,fill=qqqqff,fill opacity=0.15] (5.,0.) -- ( 7.5,-2.5) -- (5.,-5) -- (2.5,-2.5) -- cycle;
	\fill[color=qqqqff,fill=qqqqff,fill opacity=0.15] (-5.,0.) -- (-7.5,2.5) -- (-5.,5) -- (-2.5,2.5) -- cycle;
	\fill[color=qqqqff,fill=qqqqff,fill opacity=0.15] (-5.,0.) -- ( -2.5,-2.5) -- (-5.,-5) -- (-7.5,-2.5) -- cycle;
	\fill[color=qqqqff,fill=qqqqff,fill opacity=0.15] (-10.,0.) -- (-7.5,2.5) -- (-10,5) -- (-15.,5.) -- cycle;
	\fill[color=qqqqff,fill=qqqqff,fill opacity=0.15] (10.,0.) -- (7.5,2.5) -- (10,5) -- (15.,5.) -- cycle;
	\fill[color=qqqqff,fill=qqqqff,fill opacity=0.15] (10.,0.) -- (7.5,-2.5) -- (10,-5) -- (15.,-5.) -- cycle;
	\fill[color=qqqqff,fill=qqqqff,fill opacity=0.15] (-10.,0.) -- (-7.5,-2.5) -- (-10,-5) -- (-15.,-5.) -- cycle;
	\fill[color=ffqqqq,fill=ffqqqq,fill opacity=0.15] (-2.5,-2.5) -- (-5,-5) -- (0,-5) -- cycle;
	\fill[color=ffqqqq,fill=ffqqqq,fill opacity=0.15] (-2.5, 2.5) -- (-5, 5) -- (0, 5) -- cycle;
	\fill[color=ffqqqq,fill=ffqqqq,fill opacity=0.15] ( 2.5, 2.5) -- ( 5, 5) -- (0, 5) -- cycle;
	\fill[color=ffqqqq,fill=ffqqqq,fill opacity=0.15] ( 2.5,-2.5) -- ( 5,-5) -- (0,-5) -- cycle;
	\fill[color=ffqqqq,fill=ffqqqq,fill opacity=0.15] ( 7.5,-2.5) -- ( 10,-5) -- (5,-5) -- cycle;
	\fill[color=ffqqqq,fill=ffqqqq,fill opacity=0.15] ( 7.5,2.5) -- ( 10,5) -- (5,5) -- cycle;
	\fill[color=ffqqqq,fill=ffqqqq,fill opacity=0.15] ( -7.5,-2.5) -- ( -5,-5) -- (-10,-5) -- cycle;
	\fill[color=ffqqqq,fill=ffqqqq,fill opacity=0.15] ( -7.5,2.5) -- ( -5,5) -- (-10,5) -- cycle;
	\draw (0,2.5) node[qqffqq]{$\mathcal{D}$};
	\end{tikzpicture} 
	\caption{The domains $\mathcal{C}$ (conic) and $\mathcal{D}$ (diamond) are illustrated in green. The translated cones are shown in blue, and the overlappinging regions are shown in red.  \label{cone_diamond}}
\end{figure}

Consider now the case of $\sou $ sources fired simultaneously in a manner which varies in a specific way between shot locations. 
If an amplitude variation $a_n$ and a shift variation $\tau_n$ is applied to the $n$th source, the recorded data will be of the form
\begin{equation} \label{simsource}
d(t,j) = \sum_{n=1}^\sou  a_n(j) f_n(t+\tau_n(j),\triangle_xj), \quad j=-\sou m, \dots, \sou m-1,
\end{equation}
with each $f_n$ representing seismic data recorded at a certain depth, corresponding to one common receiver gather. We may without loss of generality assume that $\tau_1\equiv 0$.
As explained above, if $\mathcal{F}(f_n)$ denotes the (continuous) temporal and spatial Fourier transform of $f_n$, the support of each $\mathcal{F}(f_n)$ will be contained in the conic set
\begin{equation*}
	\mathcal{C} = \{(\omega, {k_x}): \lvert\omega\rvert>c_0  \cdot\lvert {k_x}\rvert \}\end{equation*}
	where $c_0\approx 1500\, \mathrm{m/s}$ assuming a marine environment.
Introduce the diamond shaped set
\begin{equation*}
\mathcal{D} = \mathcal{C} \setminus \{( \omega, {k_x} ) :  \lvert\omega\rvert \ge c_0 \cdot\lvert k_x  \pm 1/(\sou\triangle_x)\rvert \},
\end{equation*}
and define 
\begin{equation}\label{omega0}
\omega_0=\frac{c_0}{2\sou\triangle_x}.
\end{equation}
Then $0<\omega<2\omega_0$ when $(\omega,k_x)\in \mathcal{D}$,
and all energy is unaliased when $0<\omega<\omega_0$.
The domains $\mathcal{C}$ and $\mathcal{D}$ are depicted in Figure \ref{cone_diamond}. The widest parts of the two green diamonds in Figure \ref{cone_diamond} are located precisely at frequency $\pm \omega_0\, \mathrm{Hz}$.
We remark that using signal apparition allows for each $\mathcal{F}(f_n)$ to be perfectly reconstructed in $\mathcal{D}$, see Andersson et al.~\cite{andersson2017multisource} and Amundsen et al.~\cite{amundsen2017multisource}.
Supposing that the $\sou $ sources $f_n$ are sampled in $2\sou m$ points for some integer $m$, we define the semi-discrete Fourier transform of $f_n$ as 
\begin{equation*}
\widehat{f}_n (\omega,k)  =   \sum_{j} \Exp{\frac{-2 \pi i k j}{2\sou m}} \int_{-\infty}^{\infty}  f_n(t,\triangle_xj) \Exp{-2\pi i t \omega} \, dt . 
\end{equation*}
Using the Poisson summation formula and the condition $\supp\mathcal{F}(f_n)\subset\mathcal{C}$, it is straightforward to check that
\begin{equation}\label{fourierrelation}
\widehat{f}_n (\omega,k) =\frac{1}{\triangle_x}\mathcal{F}(f_n)\bigg(\omega,\frac{k}{2\sou m\triangle_x}\bigg)
\end{equation}
for $(\omega,k)$ such that $\lvert\omega\rvert< 2\omega_0$ and $-\sou m\le k\le \sou m-1$. This provides a relation between $\widehat{f}_n$ and $\mathcal{F}(f_n)$ in the domain of interest.

We will now present the main results of the paper in the following two theorems. The consequence of the first theorem is that the only way that $\mathcal{F}(f_n)(\omega, {k_x})$, $1\le n\le \sou $, can be uniquely determined in the diamond-shaped set $\mathcal{D}$ is if an apparition style of simultaneous source sampling is being used.

\begin{theorem}\label{maintheorem}
Suppose that data is given by \eqref{simsource}, with $\supp \mathcal{F}(f_n)\subset\mathcal{C}$ for $1 \le n \le \sou $. For $\mathcal{F}(f_n)(\omega, {k_x})$ to be uniquely determined when $(\omega, {k_x}) \in \mathcal{D}$, it is required that $a_n$ and $\tau_n$ are periodic of period $\sou$, i.e.,
\begin{equation}\label{oddeven}
a_n(j\bmod \sou)=a(j),\quad  \tau_n(j \bmod \sou )=\tau_n(j)
\end{equation}
for $1 \le n \le \sou $.
\end{theorem}

\begin{proof}
Applying a one-dimensional Fourier transform with respect to $t$ to \eqref{simsource} gives
\begin{equation}\label{t-Fourier}
\mathcal{F}_t (d)(\omega,j) = \sum_{n=1}^\sou a_n(j)\Exp{-2 \pi i \tau_n(j) \omega} \mathcal{F}_t(f_n)(\omega,\triangle_xj)
\end{equation}
for $j=-\sou m, \dots, \sou m-1$. 
For fixed $\omega$, let 
$w_n^\omega(k)$ be the discrete Fourier transform of $j\mapsto a_n(j)\Exp{2 \pi i \tau_n(j) \omega}$ evaluated at $k$, i.e.,  
\begin{equation*}
w^\omega_n(k) = \sum_j a_n(j) \Exp{-2 \pi i \left(\tau_n(j) \omega+\frac{jk}{2\sou m }\right)}.
\end{equation*}
If we apply a discrete Fourier transform to \eqref{t-Fourier}, we obtain 
\begin{equation*}
\label{conv}
\widehat{d}(\omega,k) = \sum_{n=1}^\sou  w_n^\omega \ast \widehat{f}_n(\omega,k),
\end{equation*}
where the (discrete) convolution acts on the second variable.
Let us now consider a fixed $\omega$ such that 
\begin{equation} \label{omega_condition}
\frac{2m-1}{m}\cdot\omega_0 \le \omega < 2\omega_0.
\end{equation}
In view of \eqref{fourierrelation} and the support condition $\supp\mathcal{F}(f_n)\subset\mathcal{C}$, this implies that the function $k\mapsto \widehat{f}_n(\omega,k)$ has support contained in $[-2m+1,2m-1]$. Similarly, the assumption that $\mathcal{F}(f_n)(\omega, {k_x})$ can be uniquely determined when $(\omega, {k_x}) \in \mathcal{D}$ turns into a condition of the type given in Lemma \ref{lemma1}. Hence, applying the lemma we conclude that for each $\omega$ satisfying \eqref{omega_condition}, it holds that $w_n^\omega(l)=0$ unless $l=2ml'$ where $l'$ is an integer.
In particular, the Fourier inversion formula gives
\begin{align*}
a_n(j)\Exp{2\pi i \tau_n(j)\omega}&=\frac{1}{2\sou m}\sum_{l=-\sou m}^{\sou m-1} \Exp{\frac{2\pi i l j}{2\sou m}}w_n^\omega(l)\\ 
&=\frac{1}{2\sou m}\sum_{l'=-\sou '}^{\sou '-1} \Exp{\frac{2\pi i l'j}{\sou }}w_n^\omega(2ml')
\notag\end{align*}
if $\sou =2\sou'$ is even, and
\begin{equation*}
a_n(j)\Exp{2\pi i \tau_n(j)\omega}=\frac{1}{2\sou m}\sum_{l'=-\sou '}^{\sou '} \Exp{\frac{2\pi i l'j}{\sou }}w_n^\omega(2ml')
\end{equation*}
if $\sou =2\sou'+1$ is odd.
In either case the expressions above clearly have period $\sou $, so
\begin{equation*}
a_n(j+\sou )\Exp{2\pi i \tau_n(j+\sou )\omega}=a_n(j)\Exp{2\pi i \tau_n(j)\omega}.
\end{equation*}
Taking absolute values we immediately infer that all $a_n(j)$ have period $\sou $. Taking logarithms and dividing by $2\pi i$ we then get
\begin{equation*}
(\tau_n(j+\sou)-\tau_n(j))\omega= \kappa(\omega)
\end{equation*}
for some integer-valued function $\kappa : \omega\mapsto \kappa(\omega)\in\mathbb{Z}$. However, according to \eqref{omega_condition} this has to hold for a continuous range of values $\omega$. Since the left-hand side is a continuous function of $\omega$ it follows that $\kappa$ is constant, and varying $\omega$ slightly shows that the only possible choice is $\kappa(\omega)\equiv 0$. Therefore $\tau_n$ is periodic with period $\sou $.
\end{proof}

We now present the second main result of the paper. In view of Theorem \ref{maintheorem}, the result shows that signal apparition is optimal in the sense that it maximizes the region of exact separation of multiple sources excited during seismic acquisition. In fact, using any other type of sampling method results in a strictly smaller domain of exact separation.

\begin{theorem}\label{maintheorem2}
Suppose that data is given by \eqref{simsource}, with $\supp \mathcal{F}(f_n)\subset\mathcal{C}$ for $1 \le n \le \sou $. Let $\mathcal{E}\subset\mathcal{C}$ and suppose that $\mathcal{F}(f_n)(\omega, {k_x})$ is uniquely determined when $(\omega, {k_x}) \in \mathcal{E}$. Then the area of $\mathcal{E}$ cannot be larger than the area of $\mathcal{D}$, and if the areas are equal then $\mathcal{E}=\mathcal{D}$.
\end{theorem}

\begin{proof}
As in the proof of Theorem \ref{maintheorem} we apply a semi-discrete Fourier transform to the data \eqref{simsource} and obtain
\begin{equation*}
\widehat{d}(\omega,k) = \sum_{n=1}^\sou  w_n^\omega \ast \widehat{f}_n(\omega,k),
\end{equation*}
where the (discrete) convolution acts on the second variable, and $w_n^\omega(k)$ for fixed $\omega$ is
the discrete Fourier transform of $j\mapsto a_n(j)\Exp{2 \pi i \tau_n(j) \omega}$ evaluated at $k$. We then fix $\omega$ such that 
\begin{equation}\label{omega_condition2}
\frac{m+l}{m}\cdot\omega_0 \le \omega < \frac{m+l+1}{m}\cdot\omega_0,
\end{equation}
where $0\le l\le m-1$ is arbitrary.
Using \eqref{fourierrelation} we now translate the assumption that $\mathcal{F}(f_n)(\omega, {k_x})$ can be uniquely determined when $(\omega, {k_x}) \in \mathcal{E}$ into a similar condition for the function $k\mapsto \widehat{f}_n(\omega,k)$ via
\begin{equation*}
(\omega,k/(2\sou m\triangle_x))\in\mathcal{E}\quad \Longleftrightarrow \quad k\in\Jg
\end{equation*}
for some subset of integers $\Jg\subset[-m-l,m+l]$.
One easily verifies that the conditions of Lemma \ref{lemma:cardinality} are satisfied, so the number $\Ng(l)$ of $k$ for which $\widehat{f}_n(\omega,k)$ is uniquely determined satisfies 
\begin{equation*}
\Ng(l)\le 2m-2l-1.
\end{equation*}
It is easy to see that for $\omega$ satisfying \eqref{omega_condition2} we have $(\omega,k/(2\sou m\triangle_x))\in\mathcal{D}$ if and only if $-(m-l-1)\le k\le m-l-1$, so $\Ng(l)$ is not larger than the number of $k$ such that $(\omega,k/(2\sou m\triangle_x))\in\mathcal{D}$.
Since $l$ was arbitrary this proves that $\mathcal{E}$ cannot have larger area than $\mathcal{D}$. The same argument shows that if the areas of $\mathcal{E}$ and $\mathcal{D}$ are the same, 
 it is required that  $\Ng(l)=2m-2l-1$ for each $0\le l\le m-1$. In particular, when $l=m-1$ we have $\Ng(m-1)=1$, which by Lemma \ref{lemma:Disbiggest} can only happen if $\Jg=\{0\}$. An application of Lemma \ref{lemma1} then completely determines the support of all $w_n^\omega$, and a repetition of the arguments at the end of the proof of Theorem \ref{maintheorem} shows that \eqref{oddeven} holds. Clearly, an apparition style sampling has then been used, so the region $\mathcal{E}$ in which perfect reconstruction is obtained must be equal to the diamond-shaped set $\mathcal{D}$.
\end{proof}

\section{Results}\label{results}

In this section we demonstrate the theoretical results proved above by performing numerical simulations on synthetic data, and by describing the results of a field test in which the theory was implemented.

\subsection{Numerical simulations}\label{ss:num}
We begin by demonstrating the separation of data in the diamond-shaped region $\mathcal D$ by performing two different simultaneous source experiments. In one experiment the data correspond to three simultaneous sources ($M=3$), and in the other experiment the data correspond to four simultaneous sources ($M=4$).

	\begin{figure*}
	\centering
	\begin{tikzpicture}[scale=\fsize]
	\begin{axis}[ height=8.5cm, width=\axisdefaultwidth,
	major tick length=1pt,
	xmin=-20,
	xmax=20,
	ymin=0,
	ymax=31,
	xlabel={Wavenumber ($10^{-3}/\mathrm{m}$)},
	ylabel={Frequency (Hz)},
	enlarge x limits=0.005,
	enlarge y limits=0.005,
	]
	\addplot graphics[xmin=-20,ymin=0,xmax=20,ymax=31,
	includegraphics={trim=0pt 450pt 0pt 110pt, clip}
	]{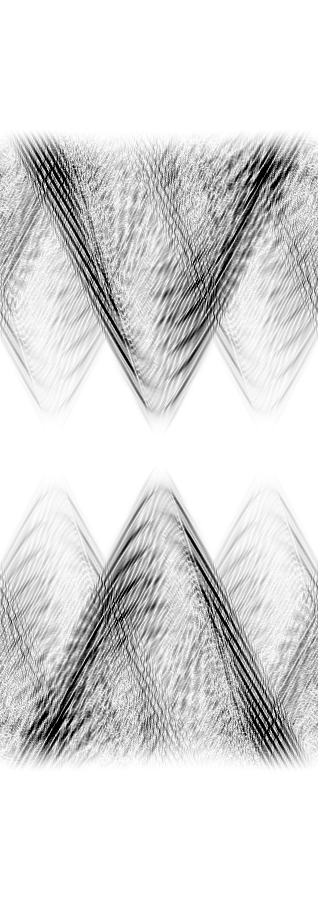};  
	\end{axis}
	\end{tikzpicture} 
	\begin{tikzpicture}[scale=\fsize]
	\begin{axis}[ height=8.5cm, width=\axisdefaultwidth,
	major tick length=1pt,
	xmin=-20,
	xmax=20,
	ymin=0,
	ymax=31,
	xlabel={Wavenumber ($10^{-3}/\mathrm{m}$)},
	ylabel={Frequency (Hz)},
	enlarge x limits=0.005,
	enlarge y limits=0.005,
	]
	\addplot graphics[xmin=-20,ymin=0,xmax=20,ymax=31,
	includegraphics={trim=0pt 450pt 0pt 110pt, clip}]{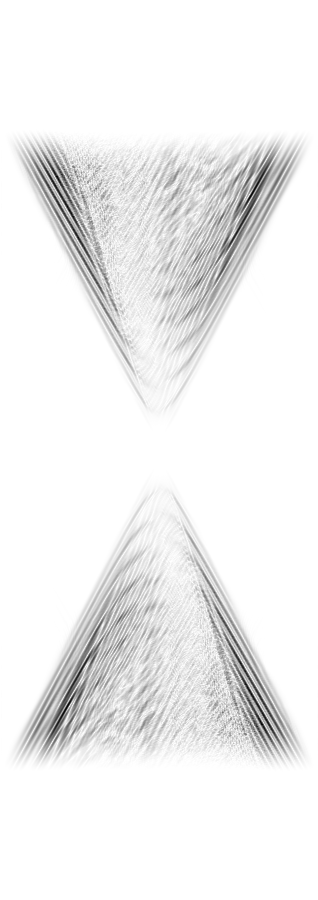};  
	\end{axis}
	\end{tikzpicture}	 
	\\
	\begin{tikzpicture}[scale=\fsize]
	\begin{axis}[ height=8.5cm, width=\axisdefaultwidth,
	major tick length=1pt,
	xmin=-20,
	xmax=20,
	ymin=0,
	ymax=31,
	xlabel={Wavenumber ($10^{-3}/\mathrm{m}$)},
	ylabel={Frequency (Hz)},
	enlarge x limits=0.005,
	enlarge y limits=0.005,
	]
	\addplot graphics[xmin=-20,ymin=0,xmax=20,ymax=31,
	includegraphics={trim=0pt 450pt 0pt 110pt, clip}]{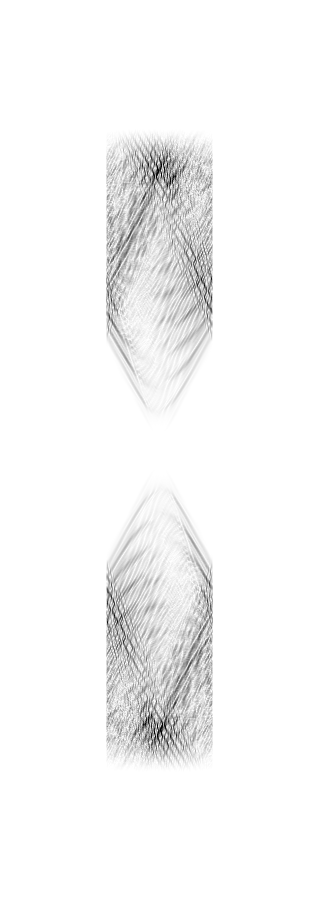};
	\end{axis}
	\end{tikzpicture} 
	\begin{tikzpicture}[scale=\fsize]
	\begin{axis}[ height=8.5cm, width=\axisdefaultwidth,
	major tick length=1pt,
	xmin=-20,
	xmax=20,
	ymin=0,
	ymax=31,
	xlabel={Wavenumber ($10^{-3}/\mathrm{m}$)},
	ylabel={Frequency (Hz)},
	enlarge x limits=0.005,
	enlarge y limits=0.005,
	]
	\addplot graphics[xmin=-20,ymin=0,xmax=20,ymax=31,
	includegraphics={trim=0pt 450pt 0pt 110pt, clip}]{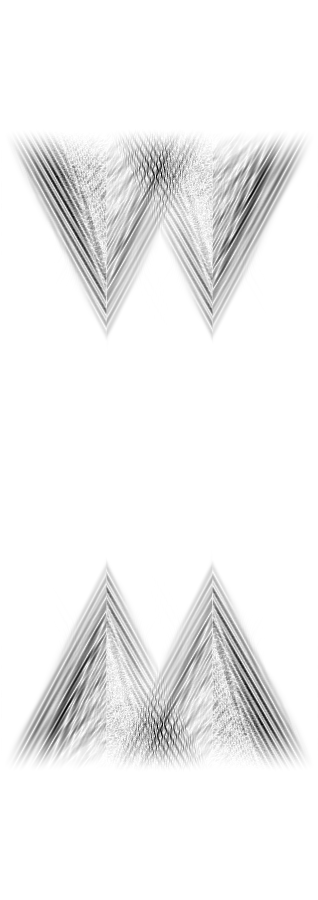};
	\end{axis}
	\end{tikzpicture}
	\caption{\label{plot_FK_app3} Frequency plots ($\omega  {k_x}$) for apparition sampling of three simultaneous sources. Blended data (top left); original source 1 (top right); reconstruction of source 1 (bottom left); reconstruction error for source 1 (bottom right).}
\end{figure*}

	\begin{figure*}
	\centering
	\begin{tikzpicture}[scale=\fsize]
	\begin{axis}[ height=8.5cm, width=\axisdefaultwidth,
	major tick length=1pt,
	xmin=-20,
	xmax=20,
	ymin=0,
	ymax=31,
	xlabel={Wavenumber ($10^{-3}/\mathrm{m}$)},
	ylabel={Frequency (Hz)},
	enlarge x limits=0.005,
	enlarge y limits=0.005,
	]
	\addplot graphics[xmin=-20,ymin=0,xmax=20,ymax=31,
	includegraphics={trim=1pt 450pt 1pt 110pt, clip}]{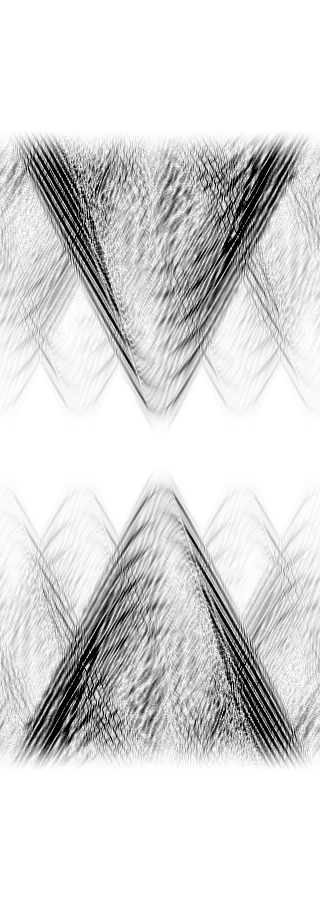};
	\end{axis}
	\end{tikzpicture} 
	\begin{tikzpicture}[scale=\fsize]
	\begin{axis}[height=8.5cm, width=\axisdefaultwidth,
	major tick length=1pt,
	xmin=-20,
	xmax=20,
	ymin=0,
	ymax=31,
	xlabel={Wavenumber ($10^{-3}/\mathrm{m}$)},
	ylabel={Frequency (Hz)},
	enlarge x limits=0.005,
	enlarge y limits=0.005,
	]
	\addplot graphics[xmin=-20,ymin=0,xmax=20,ymax=31,
	includegraphics={trim=1pt 450pt 1pt 110pt, clip}]{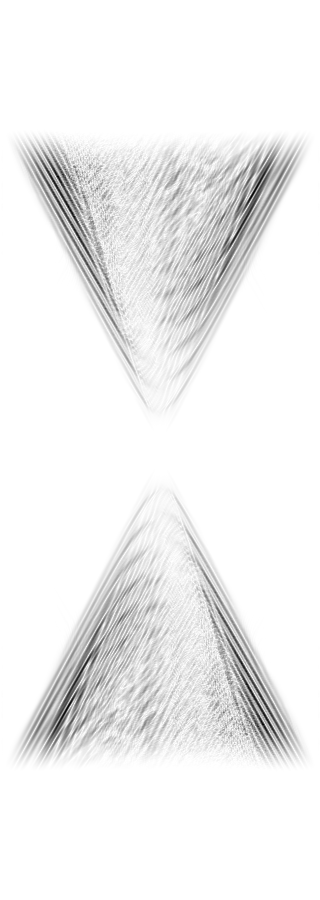};
	\end{axis}
	\end{tikzpicture} 
	\\
	\begin{tikzpicture}[scale=\fsize]
	\begin{axis}[height=8.5cm, width=\axisdefaultwidth,
	major tick length=1pt,
	xmin=-20,
	xmax=20,
	ymin=0,
	ymax=31,
	xlabel={Wavenumber ($10^{-3}/\mathrm{m}$)},
	ylabel={Frequency (Hz)},
	enlarge x limits=0.005,
	enlarge y limits=0.005,
	]
	\addplot graphics[xmin=-20,ymin=0,xmax=20,ymax=31,
	includegraphics={trim=1pt 450pt 1pt 110pt, clip}]{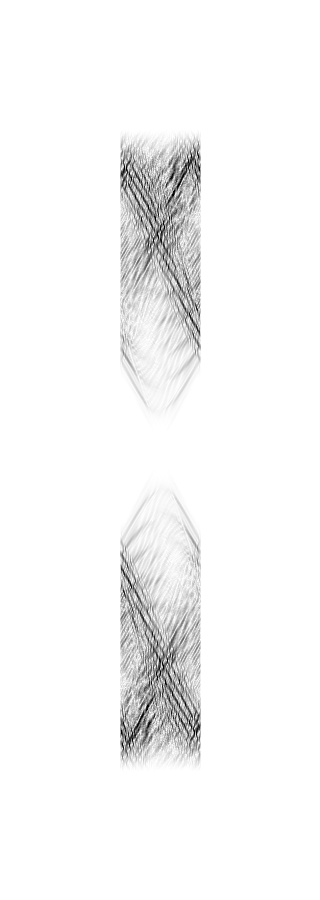};
	\end{axis}
	\end{tikzpicture} 
	\begin{tikzpicture}[scale=\fsize]
	\begin{axis}[height=8.5cm, width=\axisdefaultwidth,
	major tick length=1pt,
	xmin=-20,
	xmax=20,
	ymin=0,
	ymax=31,
	xlabel={Wavenumber ($10^{-3}/\mathrm{m}$)},
	ylabel={Frequency (Hz)},
	enlarge x limits=0.005,
	enlarge y limits=0.005,
	]
	\addplot graphics[xmin=-20,ymin=0,xmax=20,ymax=31,
	includegraphics={trim=1pt 450pt 1pt 110pt, clip}]{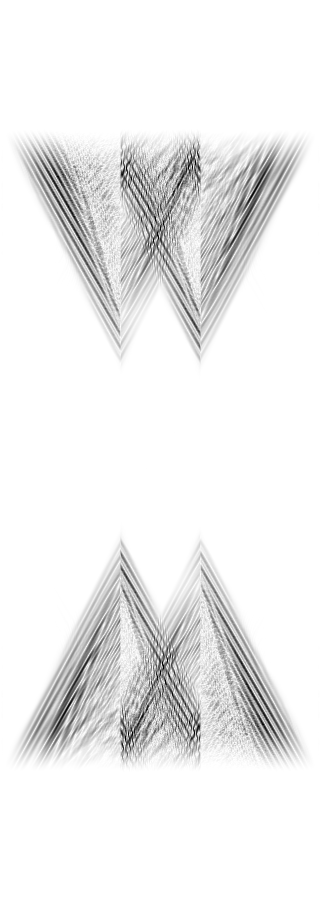};
	\end{axis}
	\end{tikzpicture}
	\caption{\label{plot_FK_app4} Frequency plots ($\omega  {k_x}$) for apparition sampling of four simultaneous sources.  Blended data (top left); original source 1 (top right); reconstruction of source 1 (bottom left); reconstruction error for source 1 (bottom right).}
\end{figure*}

We tested the method on a synthetic data set generated using an acoustic 3D finite-difference solver and a model based on salt-structures in the sub-surface and a free-surface bounding the top of the water layer.  A common receiver gather located in the middle of the model was simulated using this model in which a vessel acquires shotlines with an inline shot spacing of $25 \, \mathrm{m}$. The source wavelet comprises a Ricker wavelet with a maximum frequency of $30\, \mathrm{Hz}$.  For the apparition setup, a time delay of $\pm 24\, \mathrm{ms}$ has been applied to the data set in both experiments. The results of the experiments are displayed in Figure \ref{plot_FK_app3} ($M=3$) and Figure \ref{plot_FK_app4} ($M=4$). Both figures display blended data (top left), original source 1 (top right), reconstruction of source 1 (bottom left), and reconstruction error (bottom right). In both cases, the results clearly show that the apparition method performs perfectly within the entire diamond. Notice that the size of the diamond shrinks as the number of sources increases as predicted by equation \eqref{omega0}. The widest part of the diamond is located at $10\, \mathrm{Hz}$ ($M=3$) and at $7.5\, \mathrm{Hz}$ ($M=4$), respectively.
	
In these experiments we have only attempted to separate the signals inside the optimal diamond-shaped region $\mathcal D$ since this is the focus of the paper. Outside this region, reconstruction is often done by incorporating additional reconstruction constraints, i.e., using regularization or sparseness. There are many ways to incorporate the constraints and the reconstruction quality will depend on the data. 
However, we stress that inside $\mathcal D$, a) all the information can be recovered directly without using additional constraints, and b) the reconstruction quality will not be data dependent and the apparition method provides exact results. Furthermore, these exact results are not sensitive to random noise appearing in the periodic time-shifts \cite{wittsten2018perturbations}, \cite{wittsten2018stability}. The exact data recovered from $\mathcal D$ can subsequently be used to recover the remaining parts of the data by more elaborate de-aliasing methods, see for example Andersson et al.~\cite{andersson2016seismic}, \cite{andersson2017analytic}, \cite{andersson2017quaternion}.

\subsection{Field test}

A triple source apparition field test was carried out in the field during the summer of 2017.

A shotline with 990 shotpoints spaced $12.5\,\mathrm{m}$ apart was acquired over OBS recording stations located inline with the shotline.  The source vessel had six gunstrings which were configured into three sources spaced $30\,\mathrm{m}$ apart (center to center) with two gun strings each.  First a reference line where only the center source was fired at the desired shot locations was acquired.  The triple source line was then acquired with all three sources firing (in an apparition-style encoding pattern) and where the center source corresponds to the shot location of the reference line.  The three sources were encoded using periodic time delay sequences of $[10\ 20\ 0]\,\mathrm{ms}$ for source 1, $[10\ 10\ 10]\,\mathrm{ms}$ for source 2 and $[0\ 20\ 10]\,\mathrm{ms}$ for source 3.

Figure \ref{plot_fieldtest}a shows a common receiver gather of the triple-source data before decoding in the frequency-wavenumber domain.  The triple-source apparition encoding results in three partially overlapping cones where different mixes of the three sources are present.  These can be separated exactly (in the absence of noise and perturbations) inside the diamond-shaped regions whose geometry and optimal size are consistent with the theory presented in this paper.  The notch visible at roughly $15\,\mathrm{Hz}$ corresponds to multiples consistent with the water depth in the area.

In Figure \ref{plot_fieldtest}b we show the result after decoding Figure \ref{plot_fieldtest}a to extract the data corresponding to source 2.  As in \S\ref{ss:num}, we have focused on the diamond-shaped region only and not attempted to separate the response from the three sources outside the diamonds as this is outside the scope of the current paper.  As discussed above, an approach such as the one described by Andersson et al.~\cite{andersson2016seismic}, \cite{andersson2017analytic}, \cite{andersson2017quaternion} can be used to also decode this part of the recorded data.

Figure \ref{plot_fieldtest}c shows the same common receiver gather for the reference line.  Finally, in Figure \ref{plot_fieldtest}d we show the difference between the decoded data (Figure \ref{plot_fieldtest}b) and the reference line (Figure \ref{plot_fieldtest}c).  The outline of the diamond largely void of signal in the difference plot is clearly visible.  Bearing in mind that the triple-source and the reference lines were acquired at two different occasions with different sea and tidal states and difference in noise, source positions and other perturbations and that none of these effects have been corrected for, this demonstrates the robustness of apparition decoding.

In contrast we include the difference plots between the reference line and the decoded source 1 and source 3 in Figures \ref{plot_fieldtestcomp}a and \ref{plot_fieldtestcomp}b.  The locations of source 1 and source 3 are offset by $30\,\mathrm{m}$ in either crossline direction compared to the source locations along the reference line resulting in the large differences inside the diamond-shaped regions (compare Figures \ref{plot_fieldtestcomp}a and \ref{plot_fieldtestcomp}b with Figure \ref{plot_fieldtest}d). Finally, the same comparison in the inline direction is displayed in Figures \ref{plot_fieldtestcomp}c and \ref{plot_fieldtestcomp}d, which show the difference plots between the decoded common receiver gather 464 of source 2 (visible in Figure \ref{plot_fieldtest}b) and the reference line at adjacent receiver locations. The receiver locations of the reference line are offset $50\,\mathrm{m}$ on the sea bed in either inline direction, again resulting in large differences inside the diamond-shaped regions (compare Figures \ref{plot_fieldtestcomp}c and \ref{plot_fieldtestcomp}d with Figure \ref{plot_fieldtest}d).

\begin{figure}
	\centering
	\includegraphics[width=0.49\linewidth]{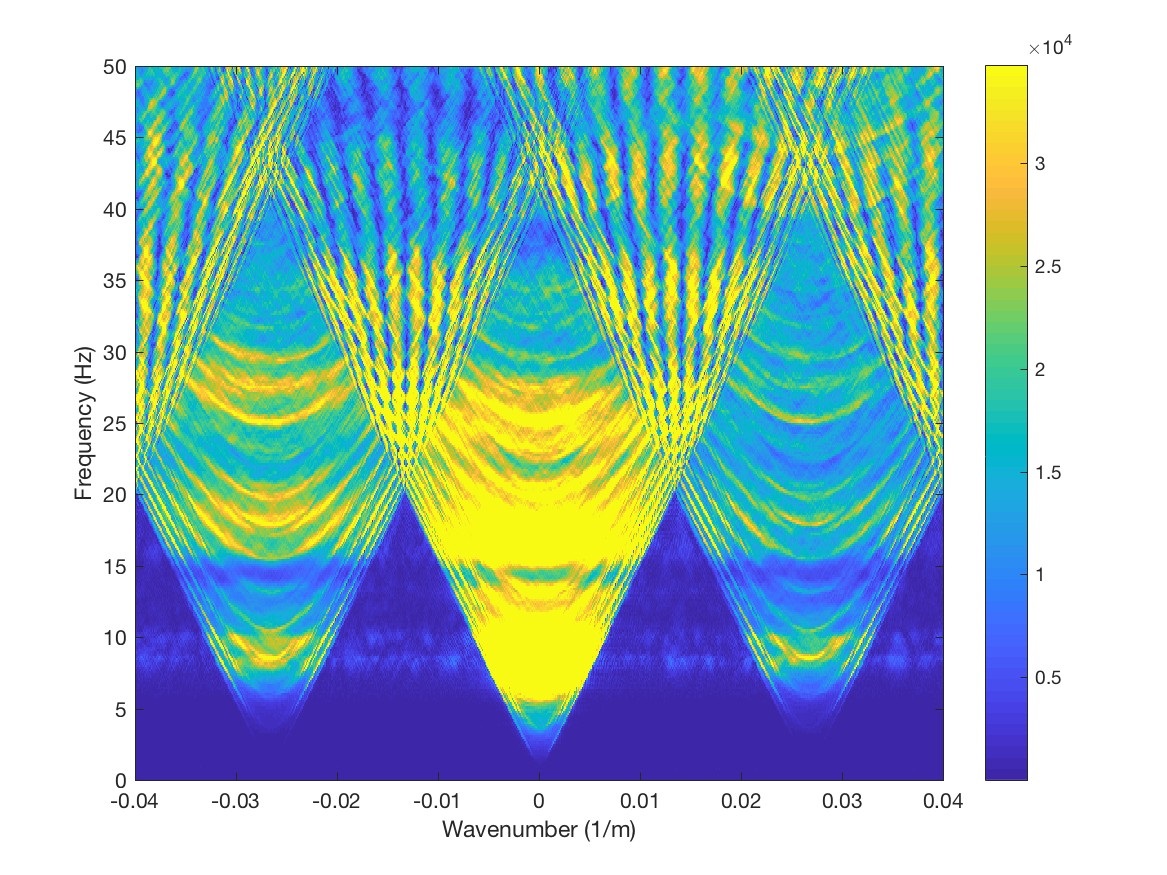}
		\includegraphics[width=0.49\linewidth]{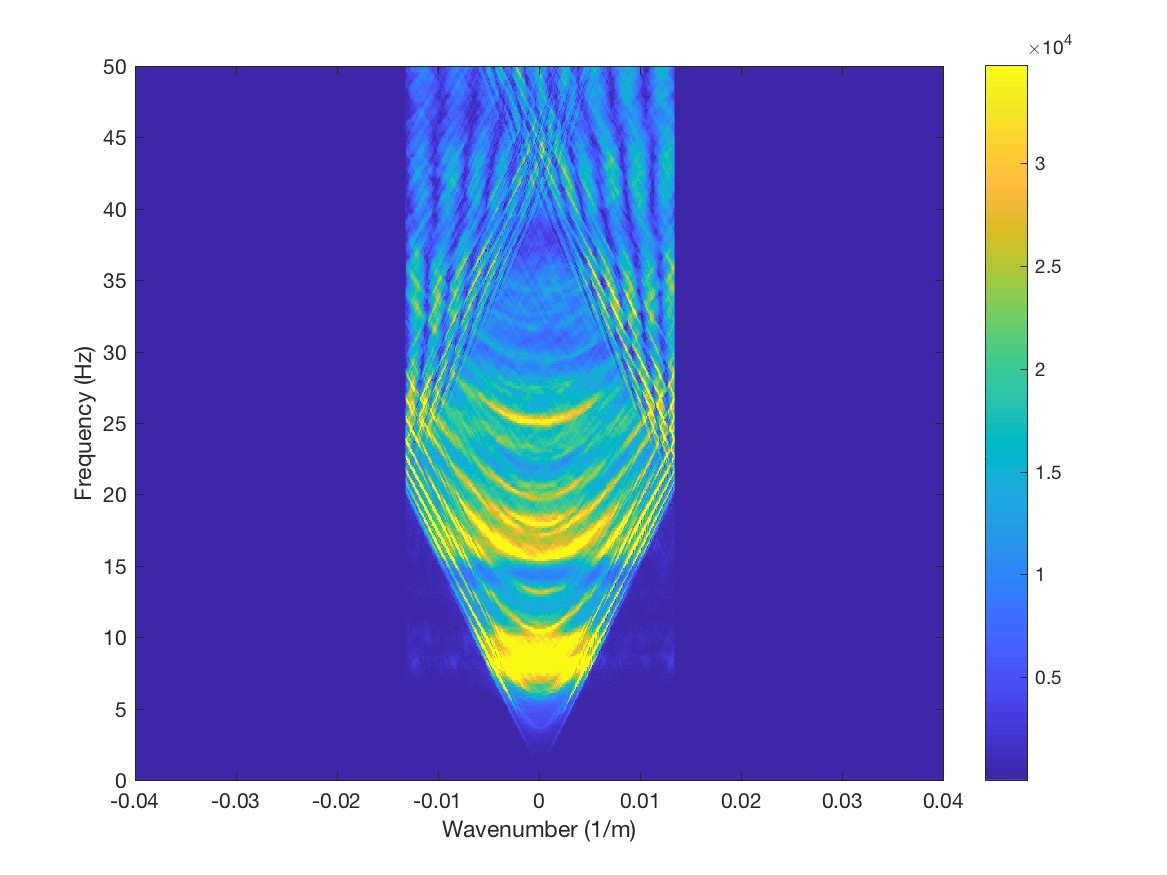}
		\\
			\includegraphics[width=0.49\linewidth]{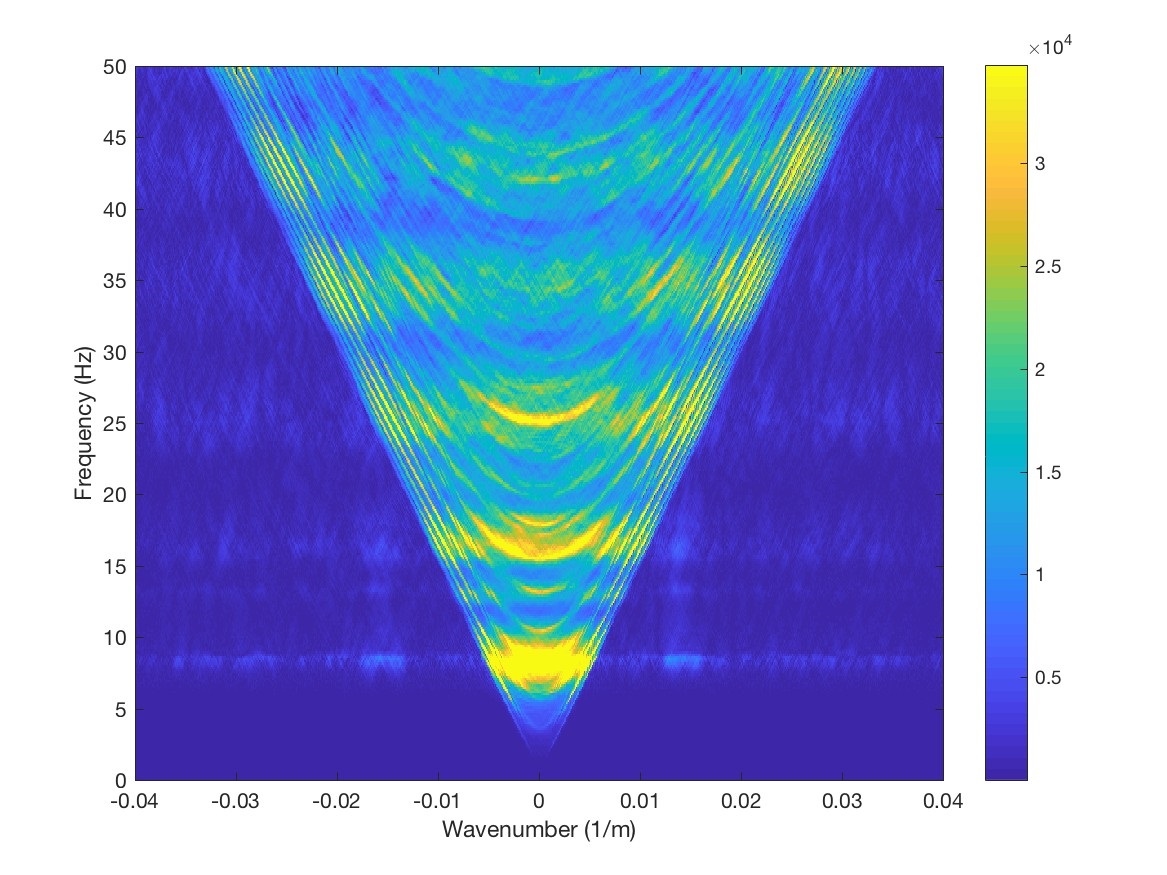}
	\includegraphics[width=0.49\linewidth]{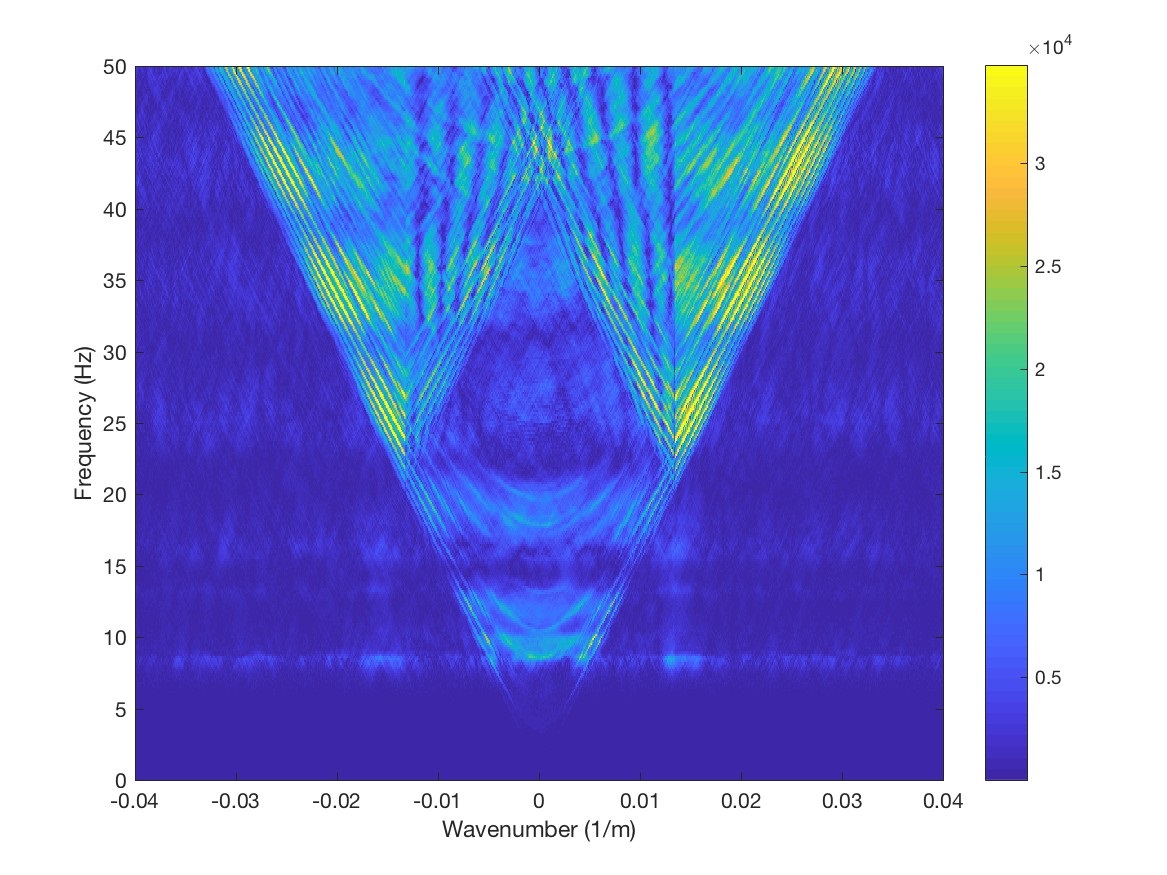}
	\caption{\label{plot_fieldtest} Real data example. Top left (a): Triple-source encoded common receiver gather. Top right (b): Decoded source 2.  Bottom left (c): Reference line corresponding to source 2.  Bottom right (d): Difference between reference line and source 2. }
\end{figure}

\begin{figure}
	\centering
			\includegraphics[width=0.49\linewidth]{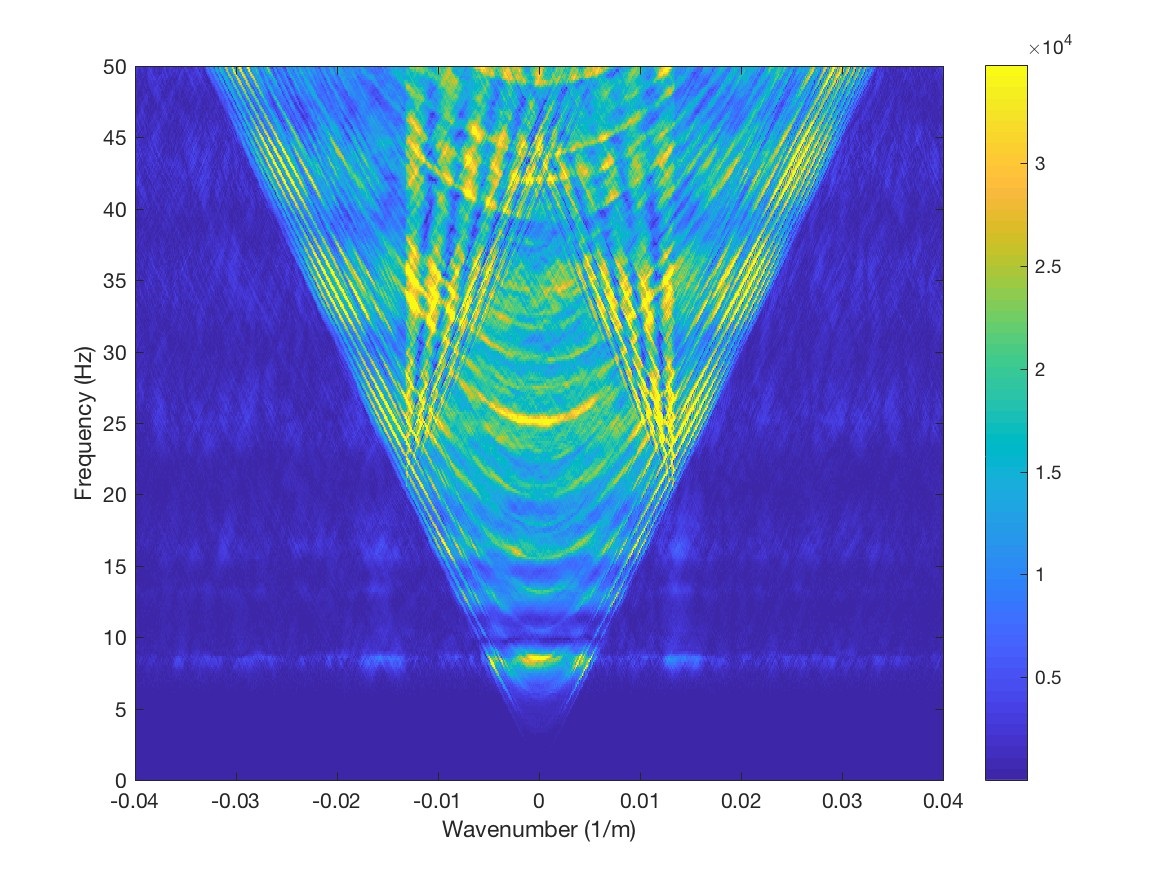}
	\includegraphics[width=0.49\linewidth]{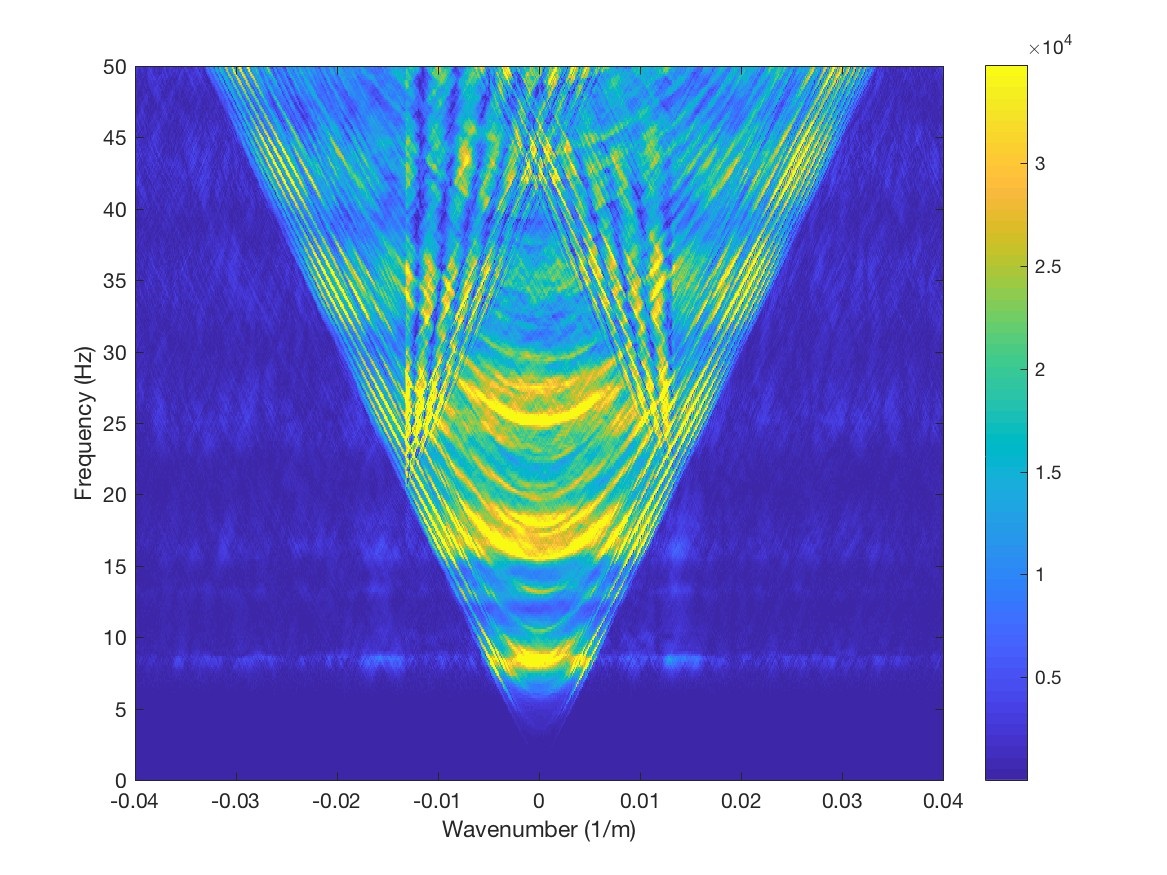}
	\\
			\includegraphics[width=0.49\linewidth]{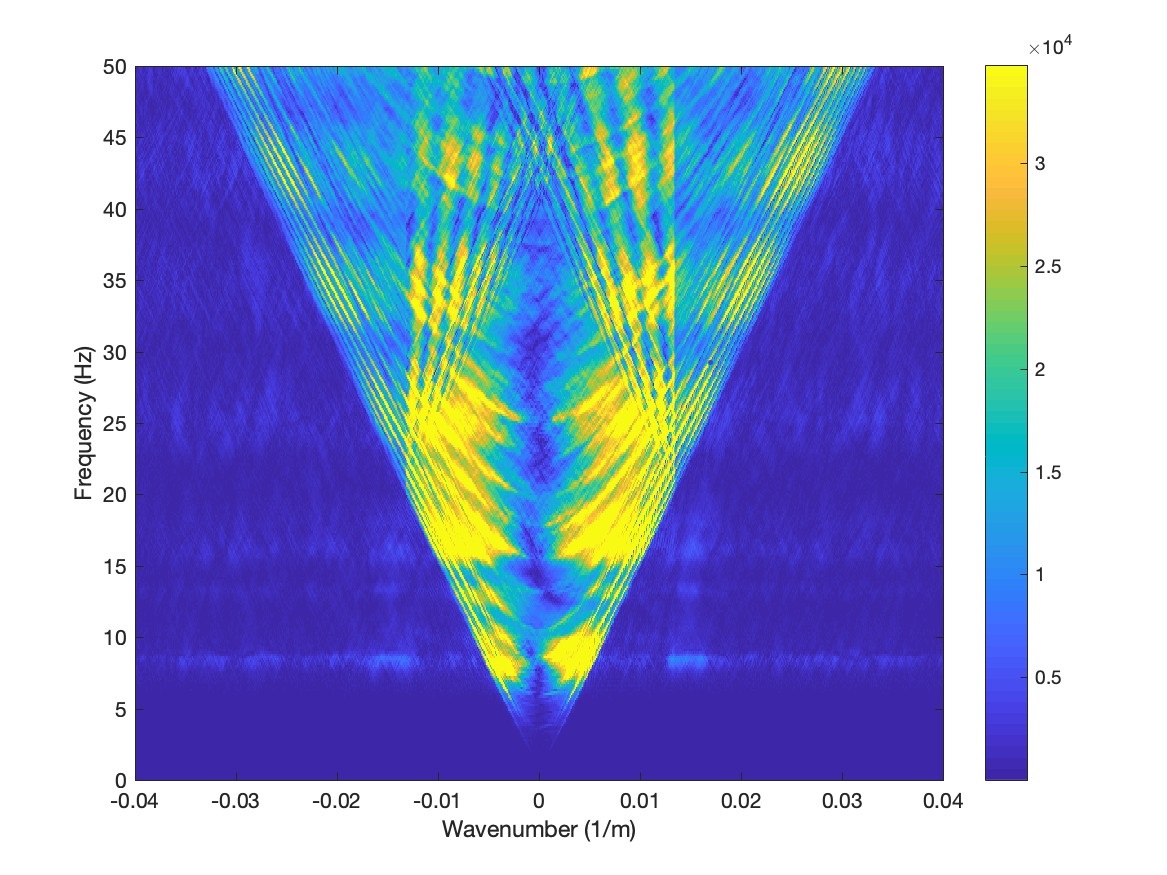}
	\includegraphics[width=0.49\linewidth]{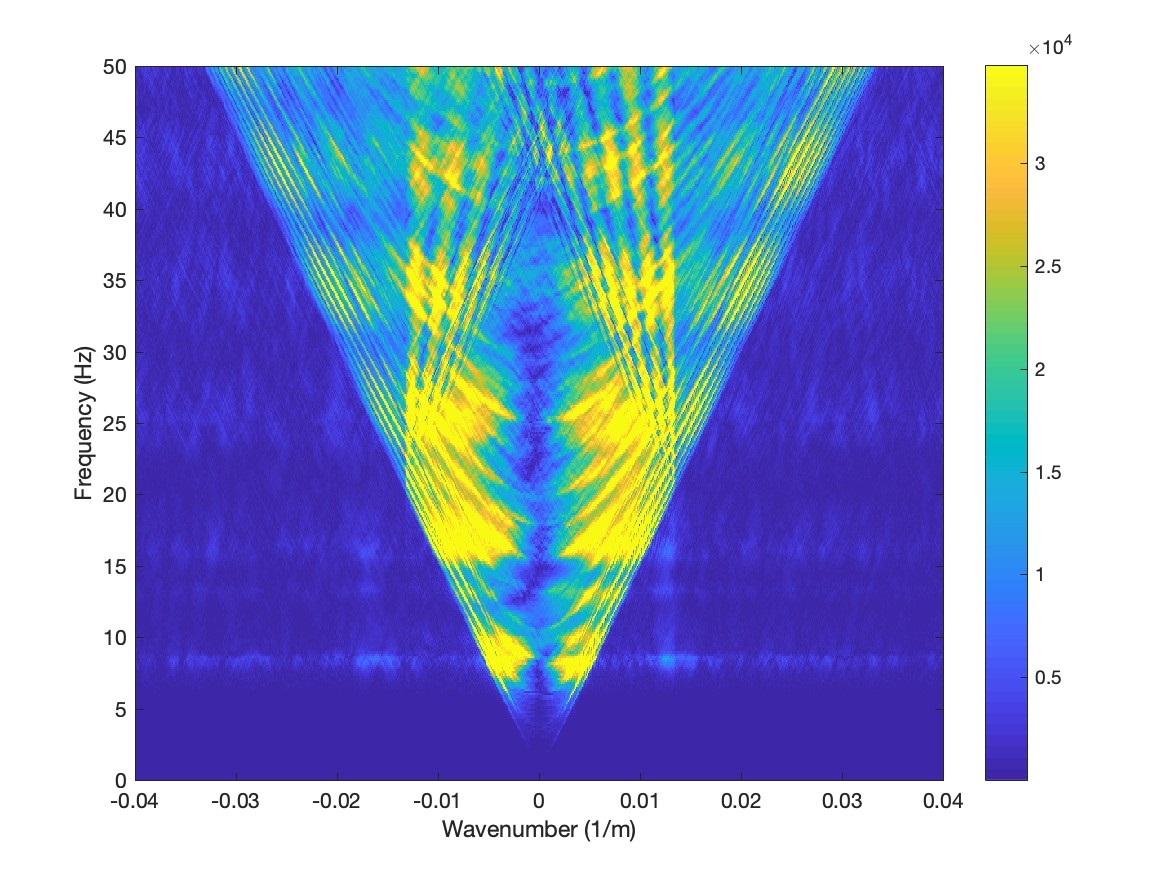}
	\caption{\label{plot_fieldtestcomp} Comparison difference plots. Top left (a): Difference between reference line and source 1.  Top right (b): Difference between reference line and source 3.  Bottom left (c): Difference between crg 464 of source 2 and crg 463 of reference line. Bottom right (d): Difference between crg 464 of source 2 and crg 465 of reference line. }
\end{figure}

\section{Discussion and conclusions}

In this paper we have presented a formal proof showing that encoding $\sou$ simultaneously excited sources in a seismic survey using the method of signal apparition results in optimally large regions in the frequency-wavenumber space for exact separation of sources.  We also presented a proof that all other methods for source encoding results in regions of exact separation which are smaller than that obtained by encoding the sources using signal apparition.

Through synthetic examples we demonstrated the exact separation of the response from three resp.~four simultaneous sources within the optimally large “flawless diamonds” in the frequency-wavenumber space. The examples show that the size of the diamond shrinks as the number of simultaneous sources increases as predicted by the theoretic results in Section \ref{section:theory}. We also provided a field test for three simultaneous sources. The field test illustrates the implementation of the method in a practical setting, and the results were consistent with the theoretical proofs of the paper. We mention that the time shifts used for three simultaneous sources in the synthetic example and the field test are of different size ($\pm 24\,\mathrm{ms}$ resp.~(the equivalent of) $\pm 10\,\mathrm{ms}$), with excellent results in both cases. Comparing the field test and the synthetic examples we also note that the flawless diamond in the field test is twice as large as the flawless diamond in the synthetic example with three sources. This is due to a more dense spacing between shotpoints (inline shot spacing of $12.5\, \mathrm{m}$ in the field test versus $25\, \mathrm{m}$ in the synthetic example) which is again consistent with the theory in Section \ref{section:theory}.

Encoding of sources using signal apparition also has other advantages. The small time-shifts (10's of $\mathrm{ms}$ for typical exploration surveys) being used enable the acquisition of excellent low-frequency content in the data.  If the sources are also fired in relatively close proximity (e.g., towed by the same vessel), the emitted wavefield will effectively be the same as if only one single source with the combined volume was fired for low frequencies.  In contrast, ``conventional'' simultaneous source acquisition using larger time-dithers common for random dithering methods (100's of $\mathrm{ms}$) often results in compromised data quality at low frequencies \cite{jiang2010analysis}, \cite{abma2012overview}.
Finally, time shifts on the order of 10’s of $\mathrm{ms}$ also result in significantly reduced peak-amplitude in source signatures and reduce output energy in the range $100\text{--}1000\, \mathrm{Hz}$ thus reducing potential negative impact on marine mammals.

\section*{Acknowledgment}
The research of Jens Wittsten was supported by Knut och Alice Wallenbergs Stiftelse.

\appendix

\section{}

In this appendix we establish a framework in which we prove auxiliary results used in the proofs of Theorems \ref{maintheorem} and \ref{maintheorem2}. 
If $\sou $ is the number of sources, and $m$ a positive integer, we shall henceforth assume that any given vector $\mathbf{f}\in\mathbb{R}^{2\sou m}$ is indexed on the range $[-\sou m,\sou m-1]$, so that
\begin{equation*}
\mathbf{f}=(\mathbf{f}(-\sou m),\ldots,\mathbf{f}(-1),\mathbf{f}(0),\mathbf{f}(1),\ldots,\mathbf{f}(\sou m-1)).
\end{equation*}
In other words, a function $f:[-\sou m,\sou m-1]\to \mathbb{R}$ is identified with its range, which is a vector in $\mathbb{R}^{2\sou m}$ denoted by $\mathbf{f}$. We let $\mathbf{e}_{-\sou m},\ldots \mathbf{e}_{\sou m-1}$ be the canonical basis of $\mathbb{R}^{2\sou m}$.

Throughout the appendix we let $l$ denote an integer $0\le l\le m-1$, and $\Jg$ some subset of integers
$\Jg\subset[-m-l,m+l]$. Define $\Jr$ by 
\begin{equation}\label{Jbad}
[-m-l,m+l]= \Jg\bigcup\Jr,
\end{equation}
and let $\Ng$ and $\Nr$ denote the cardinality (number of elements) of $\Jg$ and $\Jr$, respectively.
We then make the following standing assumption.
\begin{assumption}
The vectors $\mathbf{w}_2,\ldots,\mathbf{w}_\sou \in\mathbb{R}^{2\sou m}$ have the property that for any $\mathbf{g}_n \in \mathbb{R}^{2\sou m}$, $1\le n \le \sou $, with support contained in $[-m-l,m+l]$, it holds that the values $\mathbf{g}_n(j)$ for $j\in \Jg$ are uniquely determined by
\begin{equation*}
\mathbf{b} = \mathbf{g}_1+\mathbf{g}_2 \ast \mathbf{w}_2+\ldots+\mathbf{g}_\sou \ast \mathbf{w}_\sou 
\end{equation*}
for $1\le n \le \sou $. 
\end{assumption}

Note that for vectors $\mathbf{g}_n$ as in the statement of Assumption A, we may write
\begin{equation*}
\mathbf{g}_n\ast \mathbf{w}_n(k)=\sum_{j=-m-l}^{m+l}\mathbf{g}_n(j)\mathbf{w}_n(k-j).
\end{equation*}
We identify each such $\mathbf{g}_n$ with a vector in $\R^{2m+2l+1}$ and introduce a linear map 
\begin{equation*}
F:\mathbb{R}^{2m+2l+1}\times \ldots \times \mathbb{R}^{2m+2l+1}\to \mathbb{R}^{2\sou m}
\end{equation*}
given by	$F(\mathbf{g}_1,\ldots,\mathbf{g}_\sou )=\mathbf{b}$. Let $U$ be the linear subspace of the domain of $F$ describing the values $\mathbf{g}_n(j)$ for $j\in \Jg$ and $1\le n\le \sou $, i.e., the values which are uniquely determined by $\mathbf{b}$. 
Using the identification above we permit us to let $\mathbf{e}_{-m-l},\ldots,\mathbf{e}_{m+l}$ also denote the canonical basis of $\mathbb{R}^{2m+2l+1}$. 
For notational purposes, write $\mathbf{w}_1=\mathbf{e}_0$, so that $\mathbf{b}=\sum \mathbf{g}_n\ast \mathbf{w}_n$. 
By \cite[Lemma A.1]{andersson2017flawless} there exists a linear subspace $V=F(U)$ of $\mathbb{R}^{2\sou m}$ with
\begin{equation}\label{dim}
\dim(V)=\dim(U)=\Ng\cdot\sou 
\end{equation}
such that if $(\mathbf{g}_1,\ldots,\mathbf{g}_\sou )\in U^\perp$ then $\mathcal{P}_V(\sum \mathbf{g}_n\ast \mathbf{w}_n)=0$, where $\mathcal{P}_V$ is the orthogonal projection onto $V$.
Then all the $\sou$-tuples $(0,\ldots,0,\mathbf{e}_k,0,\ldots,0)$ with $\mathbf{e}_k$ at position $i$ belong to $U^\perp$ for $k\in J_\mathrm{bad}$ and $1\le i\le \sou $.
Likewise, the $\sou$-tuples $(0,\ldots,0,\mathbf{e}_k,0,\ldots,0)$ with $\mathbf{e}_k$ at position $i$ belong to $U$ for $k\in\Jg$ and $1\le i\le \sou $. 
This implies that 
\begin{itemize}\label{bullets}
\item $\mathbf{w}_n\ast\mathbf{e}_k\in V$ for each $n$ if $k\in\Jg$, 
\item $\mathbf{w}_n\ast\mathbf{e}_k\in V^\perp$ for each $n$ if $k\in\Jr$, and
\item $V$ is spanned by $\{\mathbf{w}_n\ast\mathbf{e}_k: k\in\Jg,\ 1\le n\le\sou\}$.
\end{itemize}
Moreover, for $1\le n_1,n_2\le \sou $ we have
\begin{equation}\label{perpagain}
\langle \mathbf{w}_{n_1}\ast \mathbf{e}_{k_1},\mathbf{w}_{n_2}\ast \mathbf{e}_{k_2}\rangle=0,\quad k_1\in\Jg,\ k_2\in\Jr.
\end{equation}

Finally, let $\mathbf{v}_1,\ldots,\mathbf{v}_\sou$ be an orthonormal basis of $\spann \{\mathbf{w}_1,\ldots,\mathbf{w}_\sou\}$ obtained from $\mathbf{w}_1,\ldots,\mathbf{w}_\sou$ by a Gram-Schmidt procedure, with $\mathbf{v}_1=\mathbf{w}_1=\mathbf{e}_0$. It is straightforward to check that 
\begin{equation}\label{span}
V=\spann\{\mathbf{v}_{n}\ast \mathbf{e}_{k}:k\in\Jg,\ 1\le n\le\sou\}.
\end{equation}
Indeed, since $\dim(V)=\Ng\cdot \sou$ it suffices to show that the set is linearly independent, but this is immediate consequence of the third bullet above.
Also, using \eqref{perpagain} it is easy to see that
\begin{equation}\label{perp2again}
\langle \mathbf{v}_{n_1}\ast \mathbf{e}_{k_1},\mathbf{v}_{n_2}\ast \mathbf{e}_{k_2}\rangle=0,\quad k_1\in\Jg,\ k_2\in\Jr
\end{equation}
for each $1\le n_1,n_2\le \sou$.

The following lemma is used in Theorem \ref{maintheorem} to show that only a signal apparition style sampling allows for perfect reconstruction in the diamond-shaped set $\mathcal{D}$.

\begin{lemma} \label{lemma1}
Let $l=m-1$ and suppose Assumption A holds. If $\Jg=\{0\}$, 
then $\supp(\mathbf{w}_n)$, $1\le n\le \sou$, is contained in the discrete set
\begin{align*}
 &\{-\sou m,-(\sou -2)m,\ldots, (\sou -2)m\} &&\text{if $\sou $ is even,}\\
 &\{-(\sou -1)m,-(\sou -3)m,\ldots, (\sou -1)m\} &&\text{if $\sou $ is odd}.
\end{align*}
In other words, $\mathbf{w}_n(k)=0$ unless $k=2ml'$ for some integer $l'$.
\end{lemma}

\begin{proof}
Clearly, it suffices to prove that each $\mathbf{v}_n$ has the desired support described in the statement of the lemma.
With $l=m-1$ and $\Jg=\{0\}$, the discussion above implies that $\mathbf{v}_n\in V$ for each $n$, that $V$ is spanned by $\{\mathbf{v}_n\}_{n=1}^\sou $, and that 
\begin{equation}\label{perp2}
\langle \mathbf{v}_{n_1}\ast\mathbf{e}_k,\mathbf{v}_{n_2}\rangle=0,\quad k=\pm 1,\ldots,\pm (2m-1)
\end{equation}
for each $1\le n_1,n_2\le \sou$. In particular we can take $n_1=n$, $n_2=1$. Since we also have $\langle \mathbf{v}_n,\mathbf{v}_1\rangle=0$ by orthogonality, this means that
\begin{equation}\label{identically0}
\mathbf{v}_n(j)=0\quad\text{for $j\in [-2m+1,2m-1]$ and $n=2,\ldots,\sou$.} 
\end{equation}

We now claim that the collection 
\begin{equation}\label{basis}
\{\mathbf{v}_n\ast \mathbf{e}_k: n=1,\ldots, \sou,\ k=1,\ldots,2m-1\}
\end{equation}
is a basis for $V^\perp$. Indeed, the cardinality of the set is $\sou (2m-1)=\dim(V^\perp)$ so the claim follows if we prove that the set is linearly independent. Arguing by contradiction, suppose that the set not linearly independent. Then for some indices $i,j$ and constants $c_{nk}$ we have 
\begin{equation*}
\mathbf{v}_i\ast\mathbf{e}_j=\sum_{(n,k)\ne(i,j)}c_{nk}\mathbf{v}_n\ast\mathbf{e}_k
\end{equation*}
with the convention that the sum is taken over $1\le n\le \sou$, $1\le k\le 2m-1$. Convolving both sides with $\mathbf{e}_{-j}$, taking the scalar product with $\mathbf{v}_i$ and using \eqref{perp2} we get
\begin{equation*}
1=\langle \mathbf{v}_i,\mathbf{v}_i\rangle=\sum_{(n,k)\ne(i,j)}c_{nk}\langle \mathbf{v}_i,\mathbf{v}_n\ast\mathbf{e}_{k-j}\rangle=0,
\end{equation*}
a contradiction.

Next, we claim that $\mathbf{e}_{2m}\in V$. To see this, recall that $\R^{2\sou m}=V\oplus V^\perp$. By the definition of direct sum, $\mathbf{e}_{2m}$ has a unique representation $\mathbf{e}_{2m}=\mathbf{a}_1+\mathbf{a}_2$ with $\mathbf{a}_1\in V$ and $\mathbf{a}_2\in V^\perp$. In view of \eqref{basis} we can write
\begin{equation*}
\mathbf{a}_2=\sum_{n=1}^\sou \sum_{k=1}^{2m-1}c_{nk}\mathbf{v}_n\ast\mathbf{e}_k.
\end{equation*}
But $\mathbf{v}_n(j)=0$ for $j=1,\ldots,2m-1$ by \eqref{identically0}, so
\begin{equation*}
\mathbf{a}_2(2m)=\sum_{n=1}^\sou \sum_{k=1}^{2m-1}c_{nk}\mathbf{v}_n(2m-k)=0.
\end{equation*}
Hence, $\mathbf{a}_1(2m)=1$, and $\mathbf{a}_1(j)=-\mathbf{a}_2(j)$ for $j\ne 2m$. By orthogonality we have
\begin{equation*}
0=\langle \mathbf{a}_1,\mathbf{a}_2\rangle=-\sum_{j\ne 2m}\mathbf{a}_2(j)^2
\end{equation*}
showing that $\mathbf{a}_2\equiv 0$, which proves the claim.

Since $\mathbf{e}_{2m}\in V$ it follows that each $\mathbf{v}_n$, $2\le n\le \sou$, satisfies
\begin{equation*}
\mathbf{v}_n(2m+k)=\langle \mathbf{v}_n,\mathbf{e}_{2m+k}\rangle=\langle \mathbf{v}_n\ast \mathbf{e}_{-k},\mathbf{e}_{2m}\rangle
=0,\quad k=1,\ldots,2m-1,
\end{equation*}
where the last identity is a consequence of \eqref{perp2}. But this means that we may now repeat the arguments in the preceding paragraph to conclude that $\mathbf{e}_{4m}\in V$, so that $\mathbf{v}_n(4m+k)=0$ for $k=1,\ldots,2m-1$. Iterating we find that $\mathbf{e}_{2ml'} \in V$ for all integers $l'$ and that $\mathbf{v}_n(2ml'+k)=0$ for $k=1,\ldots,2m-1$. Thus each $\mathbf{v}_n$ has the desired support, which completes the proof.
\end{proof}

The next two lemmas are used in Theorem \ref{maintheorem2} to show that the method of signal apparition maximizes the area of the domain in which perfect reconstruction is possible.

\begin{lemma}\label{lemma:cardinality}
Suppose that Assumption A holds.
Then the cardinality $\Ng$ of $\Jg$ is not greater than the cardinality of the discrete set $[-(m-l-1),(m-l-1)]$, i.e., $\Ng\le 2m-2l-1$.
\end{lemma}

\begin{proof}
We begin by proving the lemma for the case of $\sou=2$ sources where the ideas are easy to convey; the proof for $\sou\ge3$ sources is done in the same spirit but the arguments are more intricate then. Let $\Nr$ be the cardinality of $\Jr$. By the second bullet on page \pageref{bullets} we have $\dim(V^\perp)\ge\Nr$. Since $\Nr=2m+2l+1-\Ng$, this together with \eqref{Jbad} and \eqref{dim} implies that
\begin{equation*}
2m+2l+1-\Ng\le\dim(V^\perp)=4m-\dim(V)=4m-2\Ng.
\end{equation*}
Thus, $\Ng\le 2m-2l-1$, which completes the proof when $\sou=2$.

We now turn to the case of general $\sou\ge2$ where we will use the same ideas, namely that if there are too many elements in $\Jg$ then $\dim(V)$ will as a result be too big, forcing $\dim(V^\perp)$ to be too small. However, since the cardinality of $\Jr$ does not increase with $\sou$ while $\dim(V^\perp)$ does, 
the proof requires more finesse. We stress that it is unknown whether $0\in\Jg$.

First note that
\begin{equation*}
\dim(V^\perp)=2\sou m-\dim(V)=\sou(2m-\Ng)=\sou(\Nr-(2l+1)).
\end{equation*}
Write $N_0=\Nr-(2l+1)$ so that $\dim(V^\perp)=\sou N_0$. Assume to reach a contradiction that $\Ng>2m-2l-1$. Then $\Nr=2m+2l+1-\Ng<2(2l+1)$ so
\begin{equation*}
2N_0+1\le \Nr.
\end{equation*}
It is easy to see that this implies the existence of $N_0+1$ consecutive integers
\begin{equation*}
 j_1<j_2<\ldots<j_{N_0+1},\quad j_k\in\Jr,
\end{equation*}
such that for some $j_0\in\Z$, either $j_1+j_0$ or $j_{N_0+1}+j_0$ belongs to $\Jg$ and all perturbed elements $j_k+j_0$ belong to $[-m-l,m+l]$. We treat the case when $j_{N_0+1}+j_0\in\Jg$, the proof of the other case is similar. Consider the set
\begin{equation*}
Z=\spann\{ \mathbf{e}_{j_k}\ast \mathbf{w}_n:1\le n\le\sou,\ 1\le k\le N_0\}\subset V^\perp.
\end{equation*}
Assume first that $Z$ is linearly independent.
Since the cardinality of $Z$ is equal to $\dim(V^\perp)$, $Z$ then constitutes a basis of $V^\perp$. Hence, there are constants $c_{nk}$ (not all zero) such that
\begin{equation*}
\mathbf{e}_{j_{N_0+1}}=\sum_{n=1}^\sou\sum_{k=1}^{N_0}c_{nk}\mathbf{e}_{j_k}\ast \mathbf{w}_n.
\end{equation*}
After convolution with $\mathbf{e}_{j_0}$ we get
\begin{equation*}
\mathbf{e}_{j_0+j_{N_0+1}}=\sum_{n=1}^\sou\sum_{k=1}^{N_0}c_{nk}\mathbf{e}_{j_0+j_k}\ast \mathbf{w}_n.
\end{equation*}
By assumption, 
all terms on the right satisfy $j_0+j_k\in[-m-l,m+l]$.
But then we can choose $\mathbf{g}_n$ so that 
\begin{equation*}
\mathbf{e}_{j_0+j_{N_0+1}}=\sum_{n=1}^\sou\sum_{k=1}^{N_0}\mathbf{g}_n\ast \mathbf{w}_n=F(\mathbf{g}).
\end{equation*}
The left-hand side is also equal to $F(\mathbf{g}')$ with $\mathbf{g}'=(\mathbf{e}_{j_0+j_{N_0+1}},0,\ldots,0)$, and since $j_0+j_{N_0+1}\in\Jg$, this contradicts Assumption A.

It remains to consider the case when $Z$ is linearly dependent. Then there are constants $c_{nk}$ (not all zero) such that
\begin{equation*}
0=\sum_{n=1}^\sou\sum_{k=1}^{N_0}c_{nk}\mathbf{e}_{j_k}\ast \mathbf{w}_n.
\end{equation*}
Let $q$ be the largest integer such that $1\le q\le N_0$ and at least one $c_{nk}$ is nonzero for $k=q$. Convolving with $\mathbf{e}_{j_0+j_{N_0+1}-j_q}$ we obtain
\begin{equation*}
0=\sum_{n=1}^\sou\sum_{k=1}^{q}c_{nk}\mathbf{e}_{j_k+j_0+j_{N_0+1}-j_q}\ast \mathbf{w}_n.
\end{equation*}
Note that for $1\le k\le q$, each integer $j_k+j_0+j_{N_0+1}-j_q$ satisfies 
\begin{equation*}
-m-l\le j_k+j_0\le j_k+(j_0+j_{N_0+1}-j_q)\le j_0+j_{N_0+1}\le m+l,
\end{equation*}
and when $q=k$ we have $\mathbf{e}_{j_k+j_0+j_{N_0+1}-j_q}=\mathbf{e}_{j_0+j_{N_0+1}}\in\Jg$. But then we can choose $\mathbf{g}_n$ so that
\begin{equation*}
0=\sum_{n=1}^\sou\sum_{k=1}^{q}\mathbf{g}_n\ast \mathbf{w}_n=F(\mathbf{g}),
\end{equation*}
where $\mathcal{P}_V(F(\mathbf{g}))\ne 0$. Clearly, this is a contradiction since the projection  of the left-hand side onto $V$ is 0. This completes the proof.
\end{proof}

\begin{lemma}\label{lemma:Disbiggest}
Let $l=m-1$ and suppose that Assumption A holds. If the cardinality of $\Jg$ is equal to 1, then $\Jg=\{0\}$.
\end{lemma}

\begin{proof}
Let $j_0$ denote the single element in $\Jg$, and assume to reach a contradiction that $j_0\ne0$. By symmetry we may without loss of generality assume that $1\le j_0\le 2m-1$. Using \eqref{perp2again} it is straightforward to check that 
\begin{equation}\label{morezeros}
\mathbf{v}_n(j)=0,\quad 1\le\lvert j\rvert\le j_0+2m-1,\quad n=1,\ldots,\sou.
\end{equation}
(Note that compared to signal apparition, this is a larger range of values where each $\mathbf{v}_n$ vanishes. In particular, $\mathbf{v}_n(\pm 2m)=0$.)
We will now follow the strategy in the proof of Lemma \ref{lemma1} and
\begin{itemize}
\item[i)] determine a basis for $V^\perp$, 
\item[ii)] show that $\mathbf{e}_{2m+2j_0}\in V$, 
\item[iii)] conclude that $\mathbf{v}_n(2m+j_0+k)=0$ for $k=1,\ldots,2m-1$.
\end{itemize}
As we shall see, steps ii) and iii) may then be repeated, and iteration will lead to a contradiction showing that 
$\Jg=\{0\}$.

We begin with i) and claim that
\begin{equation}\label{basisagain}
\{\mathbf{v}_n\ast \mathbf{e}_{j_0+k}: n=1,\ldots, \sou,\ k=1,\ldots,2m-1\}
\end{equation}
is a basis for $V^\perp$. Inspecting the proof of \eqref{basis} in Lemma \ref{lemma1} we see that the same arguments can be repeated verbatim as long as we establish that
\begin{equation}\label{inVperp}
\mathbf{v}_n\ast \mathbf{e}_{j_0+k}\in V^\perp,\quad \text{for } n=1,\ldots, \sou,\ k=1,\ldots,2m-1.
\end{equation}
For $k\ge1$ this is clear as long as $j_0+k\le 2m-1$, but since $j_0\ge1$ the upper range needs to be checked. 
To this end, let $\mathbf{u}\in V$ be arbitrary. Using \eqref{span} with $\Jg=\{j_0\}$ we can write
\begin{equation*}
\langle \mathbf{u}, \mathbf{v}_{n}\ast\mathbf{e}_{j_0+k}\rangle=\sum_{n'} c_{n'}\langle  \mathbf{v}_{n'}\ast \mathbf{e}_{j_0},\mathbf{v}_{n}\ast\mathbf{e}_{j_0+k}\rangle=\sum_{n'} c_{n'}\langle  \mathbf{v}_{n'}\ast \mathbf{e}_{j_0-k},\mathbf{v}_{n}\ast\mathbf{e}_{j_0}\rangle.
\end{equation*}
By \eqref{perp2again} the right-hand side is zero when $j_0-k\in\Jr$, which holds when $k=1,\ldots,2m-1$. Hence, \eqref{inVperp} is valid, so \eqref{basisagain} is a basis for $V^\perp$. (The computation even shows that \eqref{inVperp} is true for $k=1,\ldots,2m+j_0-1$.)

Next, we prove ii). By the definition of direct sum we can write $\mathbf{e}_{2m+2j_0}=\mathbf{a}_1+\mathbf{a}_2$ uniquely with $\mathbf{a}_1\in V$ and $\mathbf{a}_2\in V^\perp$. By \eqref{basisagain} and \eqref{morezeros} we have 
\begin{align*}
\mathbf{a}_2(2m+2j_0)&=\sum_{n=1}^\sou\sum_{k=1}^{2m-1} c_{nk} \mathbf{v}_n\ast \mathbf{e}_{j_0+k}(2m+2j_0)\\&=\sum_{n=1}^\sou \sum_{k=1}^{2m-1} c_{nk} \mathbf{v}_n (2m+j_0-k)=0.
\notag\end{align*}
Hence, $\mathbf{a}_1(2m+2j_0)=1$ and $\mathbf{a}_2(j)=-\mathbf{a}_1(j)$ for all other $j$. But then 
\begin{equation*}
0=\langle \mathbf{a}_1,\mathbf{a}_2\rangle=-\sum_{j\ne 2m+2j_0}\mathbf{a}_2(j)^2
\end{equation*}
by orthogonality, so $\mathbf{a}_2\equiv 0$ which proves ii).

We now turn to iii). Since $\mathbf{e}_{2m+2j_0}\in V$ it follows that each $\mathbf{v}_n$, $2\le n\le \sou$, satisfies
\begin{equation*}
\mathbf{v}_n(2m+j_0+k)=\langle \mathbf{v}_n,\mathbf{e}_{2m+j_0+k}\rangle=\langle \mathbf{v}_n\ast \mathbf{e}_{j_0-k},\mathbf{e}_{2m+2j_0}\rangle=0
\end{equation*}
when $j_0-k\in \Jr$. In particular, 
\begin{equation*}
\mathbf{v}_n(j)=0,\quad 2m+j_0+1\le j\le 4m+2j_0-1,
\end{equation*}
which proves iii). As mentioned above, we may now repeat these arguments and conclude that $\mathbf{e}_{4m+3j_0}\in V$, so that $\mathbf{v}_n(4m+2j_0+k)=0$ for $k=1,\ldots,2m-1$. Iterating we find that $\mathbf{e}_{j} \in V$ when $j=j_0+(2m+j_0)j'$ mod $2\sou m$ for nonnegative integers $j'$ which shows that $\mathbf{v}_n(j)=0$ unless $j=(2m+j_0)j'$ mod $2\sou m$. However, taking \eqref{morezeros} into account, the number of such points in $[-\sou m,\sou m-1]$ is strictly less than $\sou$. This means that each $\mathbf{v}_n$ is a linear combination of fewer than $\sou$ number of elements $\mathbf{e}_{j}$, $j=(2m+j_0)j'$ mod $2\sou m$. Clearly, this contradicts the fact that $\dim(V)=\sou$, which proves that 
$\Jg=\{0\}$.
\end{proof}

\bibliographystyle{amsplain}
\bibliography{referenser}

\end{document}